\documentclass[12pt]{article}

\usepackage{amsmath}
\usepackage{natbib}
\usepackage{amssymb,color}
\usepackage{appendix}
\usepackage{amsthm}
\usepackage{enumerate}
\usepackage{hyperref}
\usepackage{mathtools} 
\mathtoolsset{showonlyrefs=true}  

\usepackage{caption}

\usepackage[margin=0.7in]{geometry}

\numberwithin{equation}{section}
\usepackage{comment}

\usepackage{titling}
\newtheorem{thm}{Theorem}[section]
\newtheorem{cor}[thm]{Corollary}
\newtheorem{defi}{Definition}[section]

\newtheorem{lemma}[thm]{Lemma}

\newtheorem{remark}{Remark}[section]

\begin{document}

\author{Seunghyun Lee\thanks{Department of Statistics, Columbia University, 1255 Amsterdam Avenue, New York, NY, USA. Email: sl4963@columbia.edu}  and  Hyungbin Park\thanks{Department of Mathematical Sciences, Seoul National University, 1, Gwanak-ro, Gwanak-gu, Seoul, Republic of Korea. Email: hyungbin@snu.ac.kr, hyungbin2015@gmail.com}  
}
\title{Conditions for bubbles to arise under heterogeneous beliefs}

\maketitle	

\abstract{}
This study examines the price bubble of a continuous-time asset traded in a market with heterogeneous investors. Our market constitutes of a positive mean-reverting asset and two groups of investors with different beliefs about the model parameters. We provide an equivalent condition for the formation of bubbles and show that price bubbles may not form, although there are heterogeneous beliefs. This suggests that investor heterogeneity does not always result in a bubble, and additional analysis on the drift term is required. Additionally, when investors agree on the volatility, the explicit bubble size can be calculated through a differential equation argument. Analysis of this bubble formula clearly justifies that our model is indeed realistic.

\section{Introduction}
The bubble theory is an indispensable concept that provides a rationale for rapidly increasing prices in the financial market. Price bubbles tend to have a serious impact on the economy, leading to over-investment, uncontrolled trading, and other financial crisis (e.g. \cite{xiong2013bubbles}). As bubbles such as the Internet bubble and the housing bubble continue to arise in the real world, identifying their cause has become an important focus. There have been numerous approaches to explain financial bubbles, including methods that ground their existence on heterogeneous beliefs, herding of traders, existence of stock analysts, and cultural and political changes (e.g. \cite{harrison1978speculative}, \cite{shiller2000measuring}, \cite{dass2008mutual}, \cite{bradley2008analyst}).

We build our results using the approaches of \citet{harrison1978speculative} and \citet{chen2011asset}. Harrison and Kreps were the first to explain price bubbles by heterogeneous beliefs via the analysis of an asset having finitely many dividend values. Their main argument was that investors buy an asset at a price exceeding their intrinsic valuation because of the possibility of selling it to a more optimistic buyer in the future. Although their model effectively explains the concept of bubbles, it lacks reality as they assumed a finite discrete time horizon and only finitely many dividends.

Extensive research has been conducted to modify the Harrison and Kreps model to continuous time (e.g. \citet{scheinkman2003overconfidence},  \citet{chen2011asset} and \citet{muhle2018risk}). Even though these 3 papers inherit the central idea of explaining bubbles via the possibility of trading the asset in the future (or in Economic terms, the resale option), the underlying asset and the specific model structure differ significantly.

First, Chen and Kohn extended the Harrison and Kreps model to continuous time by considering a dividend that follows the mean-reverting Ornstein--Uhlenbeck process. They assumed that heterogeneous beliefs are reflected on different mean-reversion rates and focused on the minimal equilibrium price, similar to Harrison and Kreps. They related this equilibrium price to a 2nd order non-linear differential equation and solved it to quantify the bubble. The bubble always exists and is identified as a function of the initial dividend rate.

Unlike Chen and Kohn, Scheinkman and Xiong focused the identifying the source of behavior heterogeneity, which they concluded as overconfidence. Although they also considered a dividend that follows the Ornstein--Uhlenbeck process, they assumed that the dividend is a noisy observation of an unobservable fundamental variable and considered two additional information streams (referred to as ``signals" in the original paper). They also considered trading costs, and this is reflected in the equilibrium price. Their conclusion is that investor overconfidence in different signals causes heterogeneous beliefs and leads to a bubble, which is a function of the average belief difference.

Finally, Muhle--Karbe and Nutz explained bubbles in derivative prices through a finite time equilibrium model with limited short sales. They assumed an exogenous underlying asset taking a general (uniformly elliptic) form and focused on the equilibrium price of its derivative. Similar to Chen and Kohn, they assumed that agents disagree on the model parameters. They paired the equilibrium price with investment strategies and assumed a market clearing condition. Under regularity assumptions, the price is identified as an unique solution of a Hamilton-Jacobi type PDE. The corresponding optimal strategies are closely related to the nonlinear term of the PDE and one group tends to possess the asset entirely.

Recently, new approaches have been developed to explain bubbles through the concept of holding costs and liquidity costs (e.g. \citet{muhle2020asset}). Cost were shown to influence future trading opportunities by adding penalization in terms of the expected return.

In this study, we focus on the relatively simple, yet powerful framework of \citet{chen2011asset} and investigate conditions that result in the formation of bubbles. Related literature such as \citet{scheinkman2003overconfidence} and \citet{chen2011asset} discussed a permanent bubble that exists regardless of the model parameters. \citet{muhle2018risk} implicitly mentioned the possibility of bubbles not forming, but they considered it to be extreme and assert that bubbles typically arise. In contrast, our main argument is that bubbles may not exist in general circumstances, and this condition is related to the drift term of the underlying asset (and invariant on the volatility term). The minimal equilibrium price, a concept directly related to the bubble, is identified as a viscosity weak solution of a differential equation. Additionally, under the assumption that our investors agree on the volatility, the minimal equilibrium price becomes a classical solution and the bubble size can be explicitly calculated. Even though the formula itself is complicated, we can verify that our bubble is intuitive, and reflects behaviors of the real world. Our contribution is mainly theoretical, as we claim that the existence of price bubbles under heterogeneity is not trivial and should be approached with caution.

As our model is built upon the Chen--Kohn model, it is important to discuss the differences and similarities. First, while the Chen--Kohn model assumed that the asset follows an Ornstein--Uhlenbeck process with heterogeneous mean--level beliefs, we assume a Cox--Ingersoll--Ross (CIR) process and allowed heterogeneous beliefs on all parameters. Our model improves the Chen--Kohn model by restricting the dividend rate and the price to be positive, as well as considers a more flexible notion of heterogeneity.

Second, while the Chen--Kohn model always resulted in a bubble, our model also considers the case without bubbles and classify such cases. In fact, our model generalizes the Chen--Kohn model in the sense that when the investors disagree only on the mean--level, a bubble always arises. Also, we mention that our results have two (typically different) threshold values, $\bar{D}$ and $\tilde{D}$, for the formula of the intrinsic price and equilibrium price whereas they coincide for the Chen--Kohn model. This leads to the loss of regularity of the minimal equilibrium price, as investigated throughout the paper.
 
Finally, we discuss the similarities. As the asset dynamics share the same drift term, we can easily deduce that the expected dividend rate in our model is the same as Chen--Kohn. Also, as both SDEs are time homogeneous, the price is a function of only the dividend rate and is related to similar differential equations.\footnote{However, the corresponding solutions (and consequently, the bubble size) are different due to the different structure of differential equations. Unlike Chen--Kohn, our equation is not generally solvable.} We also add a disclaimer that some concepts and mathematical methods used in this paper are based on \citet{harrison1978speculative} and \citet{chen2011asset}, for instance, the construction and existence of the minimal equilibrium price, the definition of the bubble size, and the proof of Corollary 5.3.

The remainder of this paper is organized as follows. We present our model with the necessary definitions in Section 2. In Section 3, we define the minimal equilibrium price and prove some properties. In Section 4, we consider the general case and provide our main results on the conditions under which bubbles exist. In Section 5, we consider a simpler model where investors agree on the volatility parameter, and calculate the explicit bubble size in terms of confluent hypergeometric functions. We also visualize the bubble size and provide an intuitive interpretation. Finally, in Section 6, we conclude the paper and propose directions for further research. The proof of details is provided in the appendix.

\section{Formulation}
We explore a positive mean-reverting model with heterogeneous investors. We consider two groups of investors in a continuous-time market similar to the Chen--Kohn framework. We consider only one asset: a mean-reverting dividend stream. The Chen--Kohn model assumed that the $i$th group believes that the dividend rate, $D_t$, follows the Ornstein--Uhlenbeck process
$$dD_t = \kappa_i(\theta-D_t)dt+\sigma dB_t$$
where $i=1,2$ and $\kappa_1 > \kappa_2 >0$.
However, this model has the drawback in that the dividend rates may become negative. This can cause the intrinsic value and the equilibrium value of the asset to be negative. To justify a negative price, we must assume that the asset is always possessed by one of the two groups, which is unrealistic. We resolve this issue by considering a positive price process that follows the CIR process. We assume that the $i$th group believes that the dividend rate of the asset satisfies
\begin{equation}\label{D_t}
\begin{aligned}
dD_t = \kappa_i \left( \theta_i - D_t \right) dt + \sigma_i \sqrt{D_t} dB_t.
\end{aligned}
\end{equation}
Here, $\kappa_i, \theta_i, \sigma_i$ are positive real values and $B$ is a standard Brownian motion. The parameter $\kappa_i$ is an indicator of the speed of mean-reversion, $\theta_i$ represents the long term dividend level, and $\sigma_i$ is the volatility parameter. Without the loss of generality, we will assume that $$\kappa_1 \ge \kappa_2 > 0$$ throughout this study. Note that this process is also mean-reverting and has dynamics similar to the Chen--Kohn model.

Also, note that we consider heterogeneous long-term mean levels and volatility ($\theta_i$'s and $\sigma_i$'s) for each group. These parameters are introduced to model both cases when bubbles exist and not exist. \citet{chen2011asset} show that bubbles always exist when beliefs on the mean-reversion parameter differ. However, we can easily check that the discrepancy in the mean level doesn't lead to a bubble. This results in a natural question: what relationship between the parameters creates a bubble?

In this study, we assume that there are only two groups of investors and their beliefs are different. More specifically, we assume that the heterogeneous beliefs are directly projected onto the parameters $\kappa_i, \theta_i$ and $\sigma_i$ in \eqref{D_t}. We assume that the belief of each group is constant throughout time, and cannot be altered. Furthermore, we suppose that there is no transaction cost, and our investors are risk-neutral: they calculate the present value of the asset by calculating the expectation of the discounted dividend. We also assume that $\frac{2\kappa_i\theta_i}{\sigma_i^2} \ge 1$ for both $i=1,2$. This assumption is common for the CIR model, as it almost surely precludes the dividend rate to become 0.

In the academic context, bubbles are defined as the difference between the market price and the fundamental price (e.g. \citet{scheinkman2003overconfidence}, \citet{chen2011asset}, \citet{brunnermeier2016bubbles}). For the market price, we use the equilibrium price of a market with transactions. For the fundamental price, we use the concept of intrinsic value, which is the price without any transactions. This concept is also called the ``fundamental price" or ``fundamental value" in \citet{scheinkman2003overconfidence} and \citet{jarrow2010asset}, and the exact definition of the terms differs by paper.

First, we define the ``intrinsic value" as a theoretical price of an asset determined without any transactions. For simplicity, we assume that our risk-neutral investors share a constant discount rate $\lambda$. Hence, group $i$ would evaluate the value of the asset whose dividend rate is $D$ at time $t$ by
\begin{equation}\label{intrinsic def}
\begin{aligned}
\mathbb{E}^{\mathbb{Q}_i}\left[\int_t^{\infty} e^{-\lambda (s-t)} D_sds | D_t=D\right].
\end{aligned}
\end{equation}
Here, $\mathbb{Q}_i$ is the measure associated with our dividend stream defined in \eqref{D_t}. Note that $D_t$ is a stationary process, and hence the expectation in \eqref{intrinsic def} is independent of $t$. In other words, we can simplify \eqref{intrinsic def} as
$$\mathbb{E}^{\mathbb{Q}_i}\left[\int_0^{\infty} e^{-\lambda s} D_sds | D_0=D\right].$$
Now, we are ready to define the intrinsic value.

\begin{defi}
An intrinsic value of an asset whose initial dividend rate is $D$ is defined as a function $I:\mathbb{R}^+ \to \mathbb{R}^+$ where  $\mathbb{R}^+$ is the set of positive real numbers and 
$$I(D) = \max_{i=1,2} \mathbb{E}^{\mathbb{Q}_i}\left[\int_0^{\infty} e^{-\lambda s} D_sds | D_0=D\right]$$
\end{defi}
\noindent The maximum in the definition is natural because the asset would be possessed by the group that prices it highly. As the conditional expectation of $D_t$ has an explicit formula 
\begin{equation}\label{conditional}
\begin{aligned}
\mathbb{E}^{\mathbb{Q}_i}(D_t | D_0) = D_0e^{-\kappa_i t}+\theta_i \left(1-e^{-\kappa_it}\right)
\end{aligned}
\end{equation}
(e.g. see \citet{jafari2017moments}), we can easily simplify $I(D)$. As $D_t$ is almost surely positive, we can change the expectation and integral. Now,
\begin{equation}\label{intrinsic formula}
\begin{aligned}
I(D) &= \max_{i=1,2} \Bigl\{ \int_0^{\infty} e^{-\lambda s} \mathbb{E}^{\mathbb{Q}_i}[  D_s | D_0=D]ds \Bigr\}
= \max_{i=1,2} \Bigl\{ \int_0^{\infty} e^{-\lambda s}De^{-\kappa_i s}+\theta_i \left(1-e^{-\kappa_is}\right) ds\Bigr\} \\
&= \max_{i=1,2} \Bigl\{ \frac{\theta_i}{\lambda}+\frac{1}{\lambda+\kappa_i}(D-\theta_i)\Bigr\}.
\end{aligned}
\end{equation}

\noindent An interesting fact is that this formula is exactly the same as the intrinsic value from the Ornstein--Uhlenbeck process. We can also see that the different beliefs on $\sigma_i$ do not influence the intrinsic valuation.

\eqref{intrinsic formula} can be simplified by maximizing over $i=1,2$. As we assumed $\kappa_1 \ge \kappa_2$, we calculate $I(D)$ when $\kappa_1 = \kappa_2$ and $\kappa_1 > \kappa_2$ separately.
First, when $\kappa_1 = \kappa_2$,
$$I(D) = \frac{D}{\lambda+\kappa_1}+\frac{\max(\theta_1, \theta_2) \kappa_1}{\lambda(\lambda+\kappa_1)}.$$
We can clearly see that this value is positive.
Next, when $\kappa_1 > \kappa_2$,
$$ I(D) = \begin{cases}
\frac{D}{\lambda+\kappa_1}+\frac{\theta_1\kappa_1}{\lambda(\lambda+\kappa_1)}, & \mbox{if }D\le \Bar{D} \\
\frac{D}{\lambda+\kappa_2}+\frac{\theta_2\kappa_2}{\lambda(\lambda+\kappa_2)}, & \mbox{if }D\ge \Bar{D}
\end{cases}$$
Here, $\Bar{D}:=\frac{\kappa_1\kappa_2(\theta_1-\theta_2)+\lambda (\kappa_1\theta_1 - \kappa_2\theta_2)}{\lambda (\kappa_1-\kappa_2)}$, the value $D$ that satisfies
$$\frac{D}{\lambda+\kappa_1}+\frac{\theta_1\kappa_1}{\lambda(\lambda+\kappa_1)} = \frac{D}{\lambda+\kappa_2}+\frac{\theta_2\kappa_2}{\lambda(\lambda+\kappa_2)}.$$
Note that $\Bar{D}$ can be negative if $\kappa_1\theta_1 < \kappa_2\theta_2$. In this case, $I(D)$ can be simplified as $\frac{D}{\lambda+\kappa_2}+\frac{\theta_2\kappa_2}{\lambda(\lambda+\kappa_2)}.$

Next, we consider transactions between investors, and define a market equilibrium price. Suppose that at time $t$, the dividend rate is $D_t$ and the price of the asset is $P(D_t, t)$. Assuming that group $i$ holds the asset, they can sell it at any stopping time $\tau \ge t$. Hence, the asset value evaluated by group $i$ at time $t$ would be
\begin{equation}\label{ivalue}
\begin{aligned}
\sup_{\tau \ge t} \mathbb{E}^{\mathbb{Q}_i}\left[ \int_t^{\tau} e^{-\lambda(s-t)} D_sds +e^{-\lambda(\tau-t)}P(D_{\tau},\tau) | D_t\right].
\end{aligned}
\end{equation}
At the equilibrium, the maximal valuation of the asset must be equal to its price. We can formalize this notion into the following definition.

\begin{defi}
An equilibrium price is a function $P:\mathbb{R}^+ \times \mathbb{R}^+ \to \mathbb{R}^+ $ such that
$$P(D,t) \ge I(D)$$
and
$$P(D,t) = \max_{i=1,2}\sup_{\tau \ge t} \mathbb{E}^{\mathbb{Q}_i}\left[ \int_t^{\tau} e^{-\lambda(s-t)} D_sds +e^{-\lambda(\tau-t)}P(D_{\tau},\tau) | D_t\right]$$
for all $t>0, D>0.$
\end{defi}

\noindent Because we defined the equilibrium price to be at least the intrinsic value, its positiveness is guaranteed. Note that Definition 2.2 implies that the asset would be possessed by the group with a higher assessment of \eqref{ivalue}. Owing to its implicit structure, it is impossible to characterize every equilibrium price (in fact, there are infinitely many possible prices). In this study, we focus on the minimal equilibrium price and determine whether it is strictly larger than the intrinsic value.
 
\section{The minimal equilibrium price}
In this section we focus on a specific equilibrium price, the minimal equilibrium price. As the name suggests, the minimal equilibrium price is defined as the infimum of all possible equilibrium prices. This price is natural in the sense that it is the equilibrium price obtained from a market without bubbles, due to speculation. In the next theorem, we provide an iterative construction of the minimal equilibrium price and prove its properties.

\begin{thm}
A minimal equilibrium price exists, is time-independent, and is in fact an equilibrium price. Moreover, it is continuous in $D$.
\end{thm}

\begin{proof}
First, we provide a construction of the minimal equilibrium price. This is essentially the same result as Theorem 3.1 of \citet{chen2011asset}. Hence, we will only briefly discuss the formation. For a nonnegative integer $k$, we recursively define $P_k(D, t)$ as $P_0(D,t)=I(D)$ and
\begin{equation}\label{p* def}
\begin{aligned}
P_k(D,t) = \max_{i=1,2}\sup_{\tau \ge t} \mathbb{E}^{\mathbb{Q}_i}\left[ \int_t^{\tau} e^{-\lambda(s-t)} D_sds +e^{-\lambda(\tau-t)}P_{k-1}(D_{\tau},\tau) | D_t = D \right]
\end{aligned}
\end{equation}
for $k\ge1$. This sequence of functions is nondecreasing in $k$, and 
$$P_*(D) = P_*(D,t) := \lim_{k \to \infty} P_k(D,t)$$
is time-independent and is the minimal equilibrium price. Note that we can check that $P_*$ is indeed an equilibrium price. A more detailed proof is provided in Theorem 3.1 of \citet{chen2011asset}.

Now, we prove that $P_*$ is continuous by showing $P_*$ is both lower semicontinuous and upper semicontinuous.
First, we show that $P_*(D) \le \liminf_{j\to \infty} P_*(D_j)$ for any $D_j \to D$. We fix any sequence $D_j \to D$ and $i=1,2$. We define a stopping time $\tau_j$ as follows:
$$\tau_j := \inf\left\{t\ge0 | D_t = D\right\}$$
where $D_t$ is a dividend stream with $D_0 = D_j$ that follows the belief of the $i$th investor.
Then, because $P_*$ is an equilibrium price, 
\begin{equation}
\begin{aligned}
P_*(D_j) &= \max_{i=1,2} \sup_{\tau \ge 0} \mathbb{E}^{\mathbb{Q}_i}\left[\int_0^{\tau}e^{-\lambda t}D_t dt + e^{-\lambda \tau} P_*(D_{\tau})\right]
\ge \mathbb{E}^{\mathbb{Q}_i}\left[\int_0^{\tau_j}e^{-\lambda t}D_t dt + e^{-\lambda \tau_j} P_*(D_{\tau_j})\right] \\
&= \mathbb{E}^{\mathbb{Q}_i}\left[\int_0^{\tau_j}e^{-\lambda t}D_t dt\right] +  P_*(D) \mathbb{E}^{\mathbb{Q}_i} \left[e^{-\lambda \tau_j} \right].
\end{aligned}
\end{equation}
Thus, for any $j \ge 1$,
\begin{equation}\label{lsc}
\begin{aligned}
P_*(D) \le \frac{P_*(D_j)- \mathbb{E}^{\mathbb{Q}_i}\left[\int_0^{\tau_j}e^{-\lambda t}D_t dt\right]}{\mathbb{E}^{\mathbb{Q}_i} \left[e^{-\lambda \tau_j}\right]}.
\end{aligned}
\end{equation}
Because $D_j \to D$ as $j \to \infty$, $\tau_j \to 0$ almost surely. Hence, we can apply the dominated convergence theorem and easily verify that 
$$\lim_{j \to \infty} \mathbb{E}^{\mathbb{Q}_i} \left[e^{-\lambda \tau_j }\right] = 1 \mbox{ and } \lim_{j \to \infty} \mathbb{E}^{\mathbb{Q}_i}\left[\int_0^{\tau_j}e^{-\lambda t}D_t dt\right] = 0.$$
Therefore, passing the limit in \eqref{lsc} proves that
$$P_*(D) \le \liminf_{j\to \infty} P_*(D_j)$$
and $P_*$ is lower semicontinuous.

Now, we show that $P_*(D)$ is upper semicontinuous. We fix any sequence $D_j \to D$, $i=1,2$, and consider a dividend stream $D_t$ with $D_0 = D$ such that it follows the beliefs of the $i$th investor. We define a stopping time $\tau_j'$ as
$$\tau_j' := \inf\left\{t\ge0 | D_t = D_j\right\}.$$
Similar to the previous reasoning,
\begin{equation}
\begin{aligned}
P_*(D) &= \max_{i=1,2} \sup_{\tau \ge 0} \mathbb{E}^{\mathbb{Q}_i}\left[\int_0^{\tau}e^{-\lambda t}D_t dt + e^{-\lambda \tau} P_*(D_{\tau})\right]
\ge \mathbb{E}^{\mathbb{Q}_i}\left[\int_0^{\tau_j'}e^{-\lambda t}D_t dt + e^{-\lambda \tau_j'} P_*(D_{\tau_j'})\right] \\
&= \mathbb{E}^{\mathbb{Q}_i}\left[\int_0^{\tau_j'}e^{-\lambda t}D_t dt\right] +  P_*(D_j) \mathbb{E}^{\mathbb{Q}_i}\left[e^{-\lambda \tau_j'}\right]
\end{aligned}
\end{equation}
As we take $j \to \infty$, $\tau_j' \to 0$, and we can see that 
$$P_*(D) \ge \limsup_{j\to \infty} P_*(D_j).$$
Therefore, $P_*$ is upper semicontinuous and our proof is complete.
\end{proof}

Following the notation from the proof of Theorem 3.1, we will denote $P_*(D)$ as the minimal equilibrium price. We end this section by defining the bubble size as the discrepancy between the minimal equilibrium price and the intrinsic price. In the next section, we investigate conditions for the existence of bubbles and relate the bubble size to a differential equation.

\begin{defi}
The bubble $B(D)$ is defined as $P_*(D) - I(D)$. We say that a bubble exists if $B(D) > 0$.
\end{defi}

\section{When does bubbles exist? A differential equation approach}
In Section 2, we mentioned that bubbles exist when $\theta_1 = \theta_2$ and disappear when $\kappa_1 = \kappa_2$. A natural question would be to determine the boundary conditions that lead to the formation of the bubble. In this section, we characterize equivalent conditions for bubbles to exist. First, we state a theorem that enables us to verify that a certain function is indeed an equilibrium price. Next, we provide equivalent conditions for the existence of bubbles. These conditions are highly intuitive, because they are directly related to the drift term of \eqref{D_t}. Finally, we show that the minimal equilibrium price is a viscosity supersolution of a differential equation.

\begin{thm}\label{thm4.1}
Consider the differential equation
\begin{equation}\label{de'}
\begin{aligned}
\max\{\kappa_1(\theta_1-D) \phi'+\frac{1}{2}\sigma_1^2D \phi'', \kappa_2(\theta_2-D) \phi'+\frac{1}{2}\sigma_2^2D \phi''\}-\lambda \phi + D = 0.
\end{aligned}
\end{equation}
Any positive $C^2$ solution $\phi(D)$ of \eqref{de'} with a bounded derivative is an equilibrium price.
\end{thm}

\begin{proof}
Suppose that $\phi(D)$ is a positive $C^2$ solution such that $\left| \phi'(D) \right| \le K$ for a positive constant $K$.
We show that $\phi(D)$ is an equilibrium price.
Fix any $i = 1, 2,$ and suppose that the dividend stream $D_t$ follows the belief of investor $i$.
By applying the Ito formula on $e^{-\lambda t}\phi(D_t)$,\\
\begin{equation}\label{eq1}
\begin{aligned}
d\left(e^{-\lambda t}\phi(D_t)\right)&= -\lambda e^{-\lambda t} \phi(D_t) dt + e^{-\lambda t}\phi'(D_t) dD_t + \frac{1}{2} e^{-\lambda t} \phi''(D_t) \left<dD\right>_t \\
&= e^{-\lambda t} \left[ \kappa_i(\theta_i - D_t) \phi'(D_t) + \frac{1}{2}\sigma^2 D_t \phi''(D_t) -\lambda \phi(D_t)\right]dt + \sigma e^{-\lambda t} \sqrt{D_t} \phi'(D_t)dB_t \\
&\le -e^{-\lambda t}D_tdt + \sigma e^{-\lambda t} \sqrt{D_t} \phi'(D_t)dB_t.
\end{aligned}
\end{equation}
For the final inequality, we used \eqref{de'}. Integrating \eqref{eq1} for any stopping time $\tau \ge 0$ gives
\begin{equation}
\begin{aligned}
\mathbb{E}^{\mathbb{Q}_i}\left[e^{-\lambda \tau} \phi(D_{\tau})\right] \le \phi(D_0) - \mathbb{E}^{\mathbb{Q}_i}\left[\int_0^{\tau}e^{-\lambda t}D_t dt\right] + \mathbb{E}^{\mathbb{Q}_i}\left[\int_0^{\tau} \sigma e^{-\lambda t} \sqrt{D_t} \phi'(D_t) dB_t\right].
\end{aligned}
\end{equation}
Note that the expectation of the stochastic integral is 0 because we assume that $\phi'$ is bounded. Hence, $\mathbb{E}^{\mathbb{Q}_i}\left[e^{-\lambda \tau} \phi(D_{\tau})\right] \le \phi(D_0) - \mathbb{E}^{\mathbb{Q}_i}\left[\int_0^{\tau}e^{-\lambda t}D_t dt\right]$ holds for any $i=1,2$ and $\tau \ge 0$. We can rewrite this as
\begin{equation}\label{condition2}
\begin{aligned}
\phi(D_0) \ge \max_{i=1,2} \sup_{\tau \ge 0} \mathbb{E}^{\mathbb{Q}_i}\left[\int_0^{\tau}e^{-\lambda t}D_t dt + e^{-\lambda \tau} \phi(D_{\tau})\right].
\end{aligned}
\end{equation}
Note that taking $\tau = 0$ on the right hand side of \eqref{condition2} shows that the inequality is actually an equality. Therefore, $\phi$ satisfies the second condition of the equilibrium price in Definition 2.2.

Now, we check that $\phi(D) \ge I(D)$.
As $\phi$ is positive, for $i=1,2$ and a positive integer $n$,
$$\phi(D_0) \ge \mathbb{E}^{\mathbb{Q}_{i}}\left[\int_0^{n}e^{-\lambda t}D_t dt\right].$$
By the monotone convergence theorem and the definition of $I(D_0)$, 
$$\phi(D_0) \ge \max_{i=1,2} \mathbb{E}^{\mathbb{Q}_{i}}\left[\int_0^{\infty}e^{-\lambda t}D_t dt\right] = I(D_0).$$
Therefore, $\phi$ also satisfies the first condition in Definition 2.2 and is an equilibrium price.
\end{proof}

Now, we use the aforementioned theorem to find the conditions for bubbles to exist. Our main intuition is that the drift term in the stochastic differential equation \eqref{D_t} is uniformly dominated by one group and the intrinsic value becomes an equilibrium price. 
We recall that $\kappa_1 \ge \kappa_2 > 0$ is assumed.

\begin{thm}\label{thm:main}
A bubble exists if and only if $\kappa_1 > \kappa_2$ and $\kappa_1 \theta_1 > \kappa_2 \theta_2$. 
\end{thm}

\begin{proof}
To begin with, we consider the case where $\kappa_1 = \kappa_2$ or $\kappa_1 \theta_1 \le \kappa_2 \theta_2$ and prove that there is no bubble. First, we consider the case when $\kappa := \kappa_1 = \kappa_2$ and show that $I(D)$ is an equilibrium price. Without loss of generality, we can additionally assume that $\theta_1 \ge \theta_2$. Recall that we calculated $I(D)$ in section 2 as
$$I(D) = \frac{D}{\lambda+\kappa}+\frac{\theta_1 \kappa}{\lambda(\lambda+\kappa)}.$$
Then $I$ is twice differentiable and
$$\kappa_i(\theta_i-D) I' + \frac{1}{2}\sigma_i^2D I'' = \frac{\kappa(\theta_i - D)}{\lambda+\kappa}$$
for $i=1,2$, and the maximum in \eqref{de'} is achieved at $i=1$. Hence, $I(D)$ is a solution of equation \eqref{de'}. Because $I(D)$ has a constant derivative, we can apply Theorem 4.1 and $I(D)$ is an equilibrium price. Therefore, $I(D)$ becomes the minimal equilibrium price and there is no bubble.

Next, we consider the case when $\kappa_1 \theta_1 \le \kappa_2 \theta_2$. Because the case $\kappa_1 = \kappa_2$ was discussed in the previous paragraph, we can assume that $\kappa_1 > \kappa_2$. Under these assumptions, we can check that $\bar{D} < 0$ (recall that we defined $\bar{D}$ as $\frac{\kappa_1 \kappa_2 (\theta_1 - \theta_2) + \lambda(\kappa_1 \theta_1 - \kappa_2 \theta_2)}{\lambda(\kappa_1 - \kappa_2)}$ in Section 2). Then, the intrinsic value can be written as
$$I(D) = \frac{D}{\lambda+\kappa_2}+\frac{\theta_2 \kappa_2 }{\lambda(\lambda+\kappa_2)}$$
for all $D>0$. Similar to the previous claim, the maximum in \eqref{de'} is always achieved at $i=2$ and $I(D)$ becomes a solution of the differential equation. Theorem 4.1 verifies that $I(D)$ is the minimal equilibrium price and there is no bubble.

Now, we prove that $B(D)>0$ when $\kappa_1 > \kappa_2$ and $\kappa_1 \theta_1 > \kappa_2 \theta_2$. First, we consider the case when $\Bar{D}>0$ and show that $P_*(D) > I(D)$. Using the fact that $P_*$ is an equilibrium price, we get 
\begin{equation}
\begin{aligned}
P_*(D) &\ge \max_{i=1,2} \sup_{\tau \ge 0} \mathbb{E}^{\mathbb{Q}_i}\left[\int_0^{\tau}e^{-\lambda t}D_t dt + e^{-\lambda \tau} I(D_{\tau}) | D_0 = D \right] \\
&\ge \max_{i=1,2} \mathbb{E}^{\mathbb{Q}_i}\left[\int_0^{1}e^{ -\lambda t}D_t dt + e^{-\lambda} I(D_{1}) | D_0 = D \right] \\
&> \max_{i=1,2} \mathbb{E}^{\mathbb{Q}_i}\left[ \int_0^{\infty} e^{-\lambda t}D_t dt | D_0 = D\right] = I(D).
\end{aligned}
\end{equation}
\noindent For the first inequality, we consider a constant stopping time $\tau = 1$. The final inequality holds because events $\{D_1 \le \bar{D} \}$ and $\{D_1 > \bar{D} \}$ both have positive probability and 
\begin{equation}
\begin{aligned}
\mathbb{E}^{\mathbb{Q}_i}\left[ I(D_1) | D_0 = D \right] &> \mathbb{E}^{\mathbb{Q}_i} \left[ \frac{\theta_i}{\lambda}+\frac{D_1-\theta_i}{\lambda + \kappa_i} | D_0 = D \right] \\
&= \frac{\theta_i}{\lambda}+\frac{\mathbb{E}^{\mathbb{Q}_i} \left[ D_1  | D_0 = D  \right]-\theta_i}{\lambda+\kappa_i} \\
&= \frac{\kappa_i \theta_i}{\lambda(\lambda+\kappa_i)} + \frac{(D_0-\theta_i) e^{-\kappa_i}}{\lambda + \kappa_i} \\
&= e^{\lambda} \mathbb{E}^{\mathbb{Q}_i} \left[\int_1^{\infty} e^{-\lambda t}D_t dt | D_0 = D \right].
\end{aligned}
\end{equation}

Finally, we consider the case when $\kappa_1 \theta_1 > \kappa_2 \theta_2$ and $\bar{D} \le 0$. Note that the intrinsic value becomes $I(D) = \frac{D}{\lambda + \kappa_2} + \frac{\theta_2\kappa_2}{\lambda (\lambda+\kappa_2)}$. Because $P_*$ is an equilibrium price, \\
\begin{equation}
\begin{aligned}
P_*(D) &\ge \mathbb{E}^{\mathbb{Q}_1} \left[ \int_0^{t} e^{-\lambda s}D_s ds + e^{-\lambda t} P_*(D_t) | D_0 = D \right] \\
&\ge \mathbb{E}^{\mathbb{Q}_1} \left[ \int_0^{t} e^{-\lambda s}D_s ds + e^{-\lambda t} I(D_t) | D_0 = D \right] \\
&=\int_0^t e^{-\lambda s} \mathbb{E}^{\mathbb{Q}_1} \left[D_s | D_0=D \right] ds + e^{-\lambda t} \mathbb{E}^{\mathbb{Q}_1} \left[ I(D_t) | D_0=D \right] \\
&=(D-\theta_1)\frac{1-e^{-(\lambda+\kappa_1)t}}{\lambda+\kappa_1}+\theta_1\frac{1-e^{-\lambda t}}{\lambda} +(D-\theta_1) \frac{e^{-(\lambda+\kappa_1)t}}{\lambda+\kappa_2}  + \frac{(\lambda \theta_1 + \kappa_2 \theta_2)e^{-\lambda t}}{\lambda(\lambda+\kappa_2)}
\end{aligned}
\end{equation}
for any $t>0$. Here, we used the conditional expectation formula  \eqref{conditional}.

Define $D^*(t) := \theta_1 - \frac{\kappa_2(\kappa_1 + \lambda)(\theta_2 - \theta_1)}{\lambda(\kappa_1 - \kappa_2)} \frac{1-e^{-\lambda t}}{1- e^{-(\lambda + \kappa_1)t}}$. We first note that there exists $t^* > 0$ such that $D^*(t^*) > 0.$ This claim holds because $\lim_{t \to 0} D^*(t) = \tilde{D} > 0$. Here, the assumption that $\kappa_1 \theta_1 > \kappa_2 \theta_2$ is crucial, as the existence of such $t^*$ is not guaranteed otherwise. (Indeed, for the case when $\bar{D}<0$ but $\kappa_1 \theta_1 < \kappa_2 \theta_2,$ one can show that $D^*(t) <0$ for all $t$, and the following proof would not be applicable.)

Now, we fix such $t^*$. 
For $D < D^{*}(t^*)$,
$$(D-\theta_1)\frac{1-e^{-(\lambda+\kappa_1)t^*}}{\lambda+\kappa_1}+\theta_1\frac{1-e^{-\lambda t^*}}{\lambda} +(D-\theta_1) \frac{e^{-(\lambda+\kappa_1)t^*}}{\lambda+\kappa_2}  + \frac{(\lambda \theta_1 + \kappa_2 \theta_2)e^{-\lambda t^*}}{\lambda(\lambda+\kappa_2)} > I(D),$$
or equivalently
$$(\theta_1 - D) \left[ \frac{1-e^{-(\lambda+\kappa_1) t^*}}{\lambda+\kappa_2} - \frac{1-e^{-(\lambda+\kappa_1) t^*}}{\lambda+\kappa_1}\right] > \frac{\kappa_2(\theta_2 - \theta_1)}{\lambda (\lambda+\kappa_2)} (1- e^{-\lambda t^*})$$
holds. This result implies that $P_*(D_t) > I(D_t)$ holds for $D_t < D^*(t^*)$, which happens with positive probability. Therefore,
\begin{equation}
\begin{aligned}
P_*(D) &\ge \mathbb{E}^{\mathbb{Q}_2} \left[ \int_0^{t} e^{-\lambda s}D_s ds + e^{-\lambda t} P_*(D_t) | D_0 = D \right] \\
&> \mathbb{E}^{\mathbb{Q}_2} \left[ \int_0^{t} e^{-\lambda s}D_s ds + e^{-\lambda t} I(D_t) | D_0 = D \right] \\
&= I(D)
\end{aligned}
\end{equation}
for any $D > 0$. For the final equality, we used the definition of $I(D)$ and the law of iterated expectations.
\end{proof}

Theorem 4.2 shows that there is no price bubble if one group of investor dominates the drift term $\kappa_i(\theta_i - D)$ in \eqref{D_t} uniformly in $D$. If $\kappa_1 = \kappa_2$, the group with a higher $\theta$ has a strictly dominating drift term and there are no transactions. If $\kappa_1 \theta_1 \le \kappa_2 \theta_2$, then $\kappa_1 (\theta_1 - D) < \kappa_2 (\theta_2 - D)$ holds for all $D>0$ and the second group always values the asset highly as well. As $\sigma_i$ is a parameter on the diffusion term, it has no effect on the expected asset value and has no impact on the conditions in Theorem 4.2.

Thus, we have obtained theoretical conditions that lead to a bubble. Even though the results are surprisingly simple and elegant, in practice, the size of the bubble is also important. It would be ideal to find the minimal equilibrium price in a closed form and calculate its value explicitly. It is intuitive that the minimal equilibrium price is a particular solution of the differential equation \eqref{de'}. Unfortunately, the non-linear structure of \eqref{de'} makes it difficult to find a classical solution. Furthermore, there is no guarantee of the differentiability of the minimal equilibrium price (we only know that it is continuous). Hence, we use a weaker concept of solutions to find the connection. We characterize the minimal equilibrium price as a viscosity supersolution, which is presented in the following theorem. Note that this is a general result independent of the existence of the bubble.

\begin{thm}
The minimal equilibrium price $P_*$ is a viscosity supersolution of the differential equation
\begin{equation}\label{visc}
\begin{aligned}
-\max\{\kappa_1(\theta_1-D)\phi' + \frac{1}{2}\sigma_1^2 D \phi'', \kappa_2(\theta_2-D) \phi' + \frac{1}{2}\sigma_2^2 D \phi'' \}+\lambda \phi - D = 0.
\end{aligned}
\end{equation}
\end{thm}

\begin{proof}
Take any $i=1,2$, a positive value $D_0$, and a $C^2$ function $\psi$ such that $\psi(D_0)= P_*(D_0)$ and $\psi(D) \le P_*(D)$ for all $D$. We define a dividend stream $D_t$ with the initial value $D_0$ that follows the beliefs of group $i$. Note that we are abusing notations, as $D_0$ also becomes the time 0 value of $D_t$. We also define a sequence of stopping times $\tau_n:= \frac{1}{n} \wedge \inf \{ s>0 : \lvert D_s - D_0 \rvert \ge 1 \}$ for every positive integer $n$. We apply the Ito formula on $e^{-\lambda t} \psi(D_t)$ (the expression is given in the first line of \eqref{eq1}) and integrate from 0 to $\tau_n$ to obtain
$$\mathbb{E}^{\mathbb{Q}_i}\left[e^{-\lambda \tau_n} \psi(D_{\tau_n})\right] = \psi(D_0) + \mathbb{E}^{\mathbb{Q}_i}\left[\int_0^{\tau_n} e^{-\lambda t} \left( \kappa_i(\theta_i - D_t) \psi' + \frac{1}{2}\sigma_i^2 D_t \psi'' -\lambda \psi(D_t)\right)dt \right].$$
As $\psi(D_0) = P_*(D_0) \ge \mathbb{E}^{\mathbb{Q}_i}\left[\int_0^{\tau_n} e^{-\lambda t} D_t dt + e^{-\lambda \tau_n}P_*(D_{\tau_n}) \right]$, the aforementioned equation can be written as
$$\mathbb{E}^{\mathbb{Q}_i}\left[e^{-\lambda \tau_n}\left( \psi(D_{\tau_n}) - P_*(D_{\tau_n}) \right) \right] \ge \mathbb{E}^{\mathbb{Q}_i}\left[ \int_0^{\tau_n} e^{-\lambda t} \left( \kappa_i(\theta_i - D_t) \psi' + \frac{1}{2}\sigma_i^2 D_t \psi'' -\lambda \psi(D_t) + D_t \right)dt \right].$$
Finally, we divide the aforementioned equation by $\frac{1}{n}$ and take the limit $n \to \infty$. The left hand side is less or equal to 0 by our assumption that $\psi \le P_*$. The term $\kappa_i(\theta_i - D_t) \psi' + \frac{1}{2}\sigma_i^2 D_t \psi'' -\lambda \psi(D_t) + D_t$ can be bounded by a constant independent of $n$, and we can pass the limit inside the expectation by the dominated convergence theorem. As $n \to \infty$, $\tau_n \to 0$ and
$$ \kappa_i(\theta_i-D_0)\psi' + \frac{1}{2}\sigma_i^2 D_0 \psi''- \lambda \psi(D_0) + D_0 \le 0.$$
As this is true for all $i=1, 2$, $$ -\max\{\kappa_1(\theta_1-D_0)\psi' + \frac{1}{2}\sigma_1^2 D_0 \psi'', \kappa_2(\theta_2-D_0)\psi' + \frac{1}{2}\sigma_2^2 D_0 \psi'' \}+\lambda P_*(D_0) - D_0 \ge 0$$
and $P_*$ is a viscosity supersolution of \eqref{visc}. Note that we have used the fact that $P_*$ is continuous here.
\end{proof}

Theorem 4.3 classifies the minimal equilibrium price $P_*$ as a viscosity supersolution of \eqref{visc}. One can notice that the proof of the theorem did not use the minimality of $P_*$. In fact, we can show that any lowersemicontinuous equilibrium price $Q(D)$ is also a viscosity supersolution of \eqref{visc}. This could be interpreted as a converse of Theorem 4.1. Additionally, if $P_*$ is a $C^2$ function, we can easily prove that it becomes a classical supersolution of \eqref{visc}. In fact, we can prove that every time-independent equilibrium price that is a $C^2$ function becomes a classical supersolution of \eqref{visc}. The proof follows directly from the Ito formula and we omit the details.

\section{Calculating the bubble size when investors agree on volatility}

In this section, we consider a special case when $\sigma := \sigma_1 = \sigma_2$, i.e. the investors agree on the volatility of the asset. In this case, the differential equation \eqref{de'} can be simplified into the form
\begin{equation}\label{de}
\begin{aligned}
\max\{\kappa_1(\theta_1-D), \kappa_2(\theta_2-D)\}\phi'+\frac{1}{2}\sigma^2D \phi''-\lambda \phi + D = 0.
\end{aligned}
\end{equation}
We can see that the nonlinear part is much simpler. Hence, it is easier to find classical solutions in this form. In particular, we can prove that a special classical solution can be identified as the minimal equilibrium price.

Before we elaborate on the details, we define a threshold value of $D$
$$\tilde{D}:=\frac{\kappa_1 \theta_1 - \kappa_2 \theta_2}{\kappa_1-\kappa_2},$$
under the assumption that $\kappa_1 \neq \kappa_2$. This notation will be universally used throughout this section. Note that $\tilde{D}$ is the value where $\kappa_1(\theta_1-D) = \kappa_2(\theta_2-D)$ and that 
\begin{equation}\label{tildeD}
\begin{aligned}
\max\{\kappa_1(\theta_1-D), \kappa_2(\theta_2-D)\} = 
\begin{cases}
\kappa_1(\theta_1-D) & \text{if } D<\tilde{D}\\
\kappa_2(\theta_2-D) & \text{if } D>\tilde{D}\\
\end{cases}.
\end{aligned}
\end{equation}
We can easily check that $\bar{D}$ and $\tilde{D}$ are closely related by the formula
$\bar{D} = \tilde{D} + \frac{\kappa_1 \kappa_2 (\theta_1 - \theta_2)}{\lambda (\kappa_1 - \kappa_2)}.$

Recall that we proved bubbles exist if and only if $\kappa_1 > \kappa_2$ and $\kappa_1 \theta_1 > \kappa_2 \theta_2$. We consider two sub-cases and calculate the corresponding bubble size. For the first case, we explicitly calculate the bubble size in terms of confluent hypergeometric functions. For the second case, we provide a lower bound of the bubble.

\subsection{Case 1: \texorpdfstring{$\theta_1 \le \theta_2$}{Case 1: theta 1 < theta 2}}
First, we consider the case with $\kappa_1 > \kappa_2, \theta_1 \le \theta_2, \kappa_1 \theta_1 \ge \kappa_2\theta_2$. The latter two conditions can be rewritten as $\frac{\kappa_2}{\kappa_1}\theta_2 \le \theta_1 \le \theta_2$. Our first step is to find classical solutions of \eqref{de}. In particular, we identify a specific solution that coincides with our intuition about the asset price. Trivially, the price of an asset whose current dividend is 0 must have a finite value, and it is also plausible that the price has an asymptotically linear relationship as the current dividend rate diverges to infinity. We set these beliefs as boundary conditions for \eqref{de} and prove that there is a unique solution.

\begin{thm}
$\Phi(D)$ defined as 
\begin{equation}
\begin{aligned}
\Phi(D) := \begin{cases}
\frac{D}{\lambda+\kappa_1} + \frac{\theta_1 \kappa_1}{\lambda(\lambda+\kappa_1)} + Em(D) & \text{if } D<\tilde{D}\\
\frac{D}{\lambda+\kappa_2} + \frac{\theta_2 \kappa_2}{\lambda(\lambda+\kappa_2)} + Fu(D) & \text{if } D>\tilde{D}\\
\end{cases}
\end{aligned}
\end{equation}
is the unique positive $C^2$ solution of \eqref{de} such that has a finite value at 0 and has a linear growth as $D \to\ \infty$. \\
Here $m(D):=M(a_1,b_1,x_1), u(D):=U(a_2,b_2,x_2)$, where $a_i=\frac{\lambda}{\kappa_i}, b_i=\frac{2\kappa_i\theta_i}{\sigma^2}$, $x_i=\frac{2\kappa_i D}{\sigma^2}$ for $i=1, 2$. $M$ and $U$ are confluent hypergeometric functions of the first and second type. \\
E and F are positive constants defined in the Appendix.
\end{thm}

\begin{proof}
First, we solve the differential equation \eqref{de}. We begin by considering the case when $D \le \tilde{D}$. The maximum in \eqref{de} is achieved when $i=1$, and \eqref{de} can be rewritten as
$$\kappa_1(\theta_1-D)\phi'+\frac{1}{2}\sigma^2D\phi''-\lambda\phi+D=0.$$
Define $\psi(D):=\phi(D)-\frac{D}{\lambda+\kappa_1}-\frac{\theta_1\kappa_1}{\lambda(\lambda+\kappa_1)}.$ Then $\psi(D)$ satisfies
\begin{equation}\label{eq_psi}
\kappa_1(\theta_1-D)\psi'+\frac{1}{2}\sigma^2D\psi''-\lambda\psi=0.
\end{equation}
Let $x = \frac{2\kappa_1D}{\sigma^2}$ and let $\psi(D) = \xi(x)$. Then, \eqref{eq_psi} can be rewritten as follows:
\begin{equation}\label{eq_xi}
x\xi''+\left(b_1-x\right)\xi' - a_1 \xi=0,
\end{equation}
where $a_1 = \frac{\lambda}{\kappa_1}$ and $b_1 =\frac{2\kappa_1\theta_1}{\sigma^2}$. 
\eqref{eq_xi} is a differential equation called Kummer's equation and it is known to have an explicit solution formula 
$$\xi(x) = C_1M(a_1,b_1,x) + C_2U(a_1,b_1,x)$$
(e.g. \citet{olver2010nist}).
$M(a_1,b_1,x)$ and $U(a_1,b_1,x)$ are confluent hypergeometric function of the first and second type, respectively. These functions are defined by infinite series, and the definition can be found in Chapter 13 of \citet{olver2010nist}. Using this notation, the solution of \eqref{de} can be written as 
\begin{equation}
\begin{aligned}
\phi(D)&=\xi\left(\frac{2\kappa_1D}{\sigma^2}\right)+\frac{D}{\lambda+\kappa_1}+\frac{\theta_1\kappa_1}{\lambda(\lambda+\kappa_1)} \\
&= C_1M\left(a_1,b_1,\frac{2\kappa_1D}{\sigma^2}\right) + C_2U\left(a_1,b_1,\frac{2\kappa_1D}{\sigma^2}\right)+\frac{D}{\lambda+\kappa_1}+\frac{\theta_1\kappa_1}{\lambda(\lambda+\kappa_1)}
\end{aligned}
\end{equation}
for $D \le \tilde{D}$. Similarly, the solution when $D > \tilde{D}$ is
$$\phi(D)=C_3M\left(a_2,b_2,\frac{2\kappa_2D}{\sigma^2}\right) + C_4U\left(a_2,b_2,\frac{2\kappa_2D}{\sigma^2}\right)+\frac{D}{\lambda+\kappa_2}+\frac{\theta_2\kappa_2}{\lambda(\lambda+\kappa_2)},$$
where $a_2 = \frac{\lambda}{\kappa_2}$ and $b_2 =\frac{2\kappa_2\theta_2}{\sigma^2}$. 

Now, we find appropriate constants, $C_1, C_2, C_3, C_4$ that constrains $\phi(D)$ to have a finite value at 0 and makes it linear near infinity. 
First, observe that $M(a,b,x) = 1+O(x)$ and $U(a,b,x) = \frac{\Gamma(1-b)}{\Gamma(a+1-b)}x^{1-b}+O(x^{2-b})$ as $x \to 0$ (see Chapter 13.2 of \citet{olver2010nist}). Because we have assumed that $b_1 = \frac{2\kappa_1\theta_1}{\sigma^2} \ge 1$ in Section 2, $U\left(a_1,b_1, \frac{2\kappa_1D}{\sigma^2}\right)$ diverges and $C_2$ must be 0.

Next, note that $M(a,b,x) = O(\frac{e^x x^{a-b}}{\Gamma(a)})$ and $U(a,b,x) = O(x^{-a})$ as $x \to \infty$ (again, see Chapter 13.2 of \citet{olver2010nist}). Regardless of $a_2$ and $b_2$, $M\left(a_2,b_2, \frac{2\kappa_2D}{\sigma^2}\right)$ diverges at an exponential rate; therefore, $C_3$ must be 0. Because $U\left(a_2,b_2, \frac{2\kappa_2D}{\sigma^2}\right)$ converges to 0, we can see that $C_4U\left(a_2,b_2, \frac{2\kappa_2D}{\sigma^2}\right)$ converges to 0 for any $C_4$, and $\phi(D)$ is linear at infinity.
Hence, $\phi(D)$ can be simplified as follows, using the notation $m(D)=M(a_1,b_1,x_1)$ and $u(D)=U(a_2,b_2,x_2)$.
\begin{equation}
\begin{aligned}
\phi(D) = \begin{cases}
\phi_1(D) := \frac{D}{\lambda+\kappa_1} + \frac{\theta_1 \kappa_1}{\lambda(\lambda+\kappa_1)} + C_1 m(D) &\text{if } D<\tilde{D}\\
\phi_2(D) := \frac{D}{\lambda+\kappa_2} + \frac{\theta_2 \kappa_2}{\lambda(\lambda+\kappa_2)} + C_4 u(D) &\text{if } D>\tilde{D}
\end{cases}
\end{aligned}
\end{equation}

Finally, we uniquely determine the values $C_1$ and $C_4$ that make $\phi$ a $C^2$ function. We need $\phi_1(\tilde{D}) = \phi_2(\tilde{D})$, $\phi_1'(\tilde{D}) = \phi_2'(\tilde{D})$, and $\phi_1''(\tilde{D}) = \phi_2''(\tilde{D})$. We focus on the first two equations. Using the formula 
\begin{equation}\label{diff.m}
\frac{\partial M(a,b,x)}{\partial x} = \frac{a}{b}M(a+1,b+1,x)
\end{equation}
and
\begin{equation}\label{diff.u}
\frac{\partial U(a,b,x)}{\partial x} = -aU(a+1,b+1,x)
\end{equation}
from \citet{olver2010nist}, we can uniquely determine $C_1$ and $C_4$. The linear independence of the two equations follows from a simple observation of the coefficients' sign. We omit the details, and the solutions, $C_1, C_4$, are defined as $E$ and $F$ in Appendix. There, we show that $E$ and $F$ are both positive. The relationship, $\phi_1''(\tilde{D}) = \phi_2''(\tilde{D})$, follows from our assumption that $\phi_1(\tilde{D}) = \phi_2(\tilde{D})$ and $\phi_1'(\tilde{D}) = \phi_2'(\tilde{D})$ by taking limits $D \to \tilde{D}+$ and $D \to \tilde{D}-$ in \eqref{de}. Therefore, $E, F$ are the unique constants that make $\phi$ a $C^2$ function, and our proof is complete.
\end{proof}

One may question whether $\Phi(D)$ is increasing. This property must hold to be analogous to our intuition. The following lemma confirms this; it additionally shows that $\Phi'(D)$ is bounded above.

\begin{lemma}
$\Phi(D)$ is convex and $\frac{1}{\lambda+\kappa_1} < \Phi'(D) < \frac{1}{\lambda+\kappa_2}$.
\end{lemma}

\begin{proof}
Recall that we defined $\Phi(D)$ as 
\begin{equation}
\begin{aligned}
\Phi(D) = \begin{cases}
\frac{D}{\lambda+\kappa_1} + \frac{\theta_1 \kappa_1}{\lambda(\lambda+\kappa_1)} + Em(D) & \text{if } D<\tilde{D}\\
\frac{D}{\lambda+\kappa_2} + \frac{\theta_2 \kappa_2}{\lambda(\lambda+\kappa_2)} + Fu(D) & \text{if } D>\tilde{D}\
\end{cases}.
\end{aligned}
\end{equation}
We can easily verify that $m(D)$ and $u(D)$ are convex using the differential formulas \eqref{diff.m} and \eqref{diff.u}. Therefore, $\Phi(D)$ is also convex, and $\Phi'(D)$ is increasing in $D$. Note that
\begin{equation}
\begin{aligned}
\Phi'(D) = \begin{cases}
\frac{1}{\lambda+\kappa_1} + \frac{2\kappa_1}{\sigma^2} \frac{a_1}{b_1} M\left(a_1+1,b_1+1,\frac{2\kappa_1 D}{\sigma^2} \right) E & \text{if } D<\tilde{D}\\
\frac{1}{\lambda+\kappa_2} - \frac{2\kappa_2}{\sigma^2} a_2 U\left(a_2+1,b_2+1,\frac{2\kappa_2 D}{\sigma^2} \right) F & \text{if } D>\tilde{D}
\end{cases}
\end{aligned}
\end{equation}
and 
$$\lim_{D \to 0} M\left(a_1+1,b_1+1,\frac{2\kappa_1 D}{\sigma^2} \right) = 1, \lim_{D \to \infty} U\left(a_2+1,b_2+1,\frac{2\kappa_2 D}{\sigma^2} \right) = 0$$
(here, we used the asymptotic properties discussed in the proof of the Theorem 5.1).
Hence, 
$$\lim_{D \to 0} \Phi'(D) = \frac{1}{\lambda+\kappa_1} + \frac{2\kappa_1}{\sigma^2} \frac{a_1}{b_1}, \lim_{D \to \infty} \Phi'(D) = \frac{1}{\lambda+\kappa_2}$$ and $\frac{1}{\lambda+\kappa_1} < \Phi'(D) < \frac{1}{\lambda+\kappa_2}$ follows from the convexity of $\Phi$.
\end{proof}

Finally, in Corollary 5.3, we use the previous theorems and lemmas to prove that $\Phi(D)$ is indeed the minimal equilibrium price. This simple corollary identifies the ``minimal equilibrium price" as the classical solution of \eqref{de'}. 

\begin{cor}
$\Phi(D)$, as defined in Theorem 5.1, is the minimal equilibrium price. In other words, $\Phi(D) = P_*(D)$.
\end{cor}

\begin{proof}
According to Theorem 5.1 and Lemma 5.2, $\Phi$ is a positive $C^2$ solution of \eqref{de} (and also \eqref{de'}) that has bounded derivatives. Hence, we can apply Theorem 4.1 and $\Phi$ is an equilibrium price. By the minimality of $P_*(D)$, we know that $\Phi(D) \ge P_*(D)$. We prove that $\Phi(D) \le P_*(D)$ by considering the minimization problem 
\begin{equation}\label{inf}
\inf_{D\ge0} \{ P_*(D)-\Phi(D) \}.
\end{equation}
This approach was originally mentioned in \citet{chen2013erratum}.

First, we consider the case when the infimum is achieved as $D \to \infty$. Note that $0 \le \Phi(D)-P_*(D) \le \Phi(D)-I(D)=FU(D)$ for a sufficiently large $D$ and $U(D) \to 0$ as $D \to \infty$. Hence, $\Phi(D)-P_*(D) \to 0$ as $D \to \infty$ and the infimum of \eqref{inf} is 0.

Next, we consider the case when the infimum of \eqref{inf} is achieved at a finite value, $D_1$. We show that $P_*(D_1) \ge \Phi(D_1)$, which completes the proof. For simplicity, define $F(D,d,p,\gamma)$ as
\begin{equation}
\begin{aligned}
F(D,d,p,\gamma) := -\max \{\kappa_1(\theta_1-D)p+\frac{1}{2}\sigma^2D \gamma, \kappa_2(\theta_2-D)p+\frac{1}{2}\sigma^2D \gamma \}+\lambda d - D.
\end{aligned}
\end{equation}
The definition of $D_1$ makes it a global minimum of $P_*-\Phi$. Because $P_*$ is a viscosity supersolution of \eqref{visc}, $F(D_1, P_*(D_1), \Phi'(D_1), \Phi''(D_1))\ge 0$. In other words,
\begin{equation}\label{eq11}
\begin{aligned}
-\max_{i=1,2}\{\kappa_1(\theta_1-D_1)\Phi'+ \frac{1}{2}\sigma^2D_1 \Phi'', \kappa_2(\theta_2-D_1)\Phi'+ \frac{1}{2}\sigma^2D_1 \Phi''\} +\lambda P_*(D_1) - D_1 \ge 0.
\end{aligned}
\end{equation}
However, because $\Phi$ is a classical solution of \eqref{de'}, 
\begin{equation}\label{eq22}
\begin{aligned}
-\max_{i=1,2}\{\kappa_1(\theta_1-D_1)\Phi'+ \frac{1}{2}\sigma^2D_1 \Phi'', \kappa_2(\theta_2-D_1)\Phi'+ \frac{1}{2}\sigma^2D_1 \Phi'' \} +\lambda \Phi(D_1) - D_1 = 0.
\end{aligned}
\end{equation}
By comparing \eqref{eq11} and \eqref{eq22}, it is clear that $P_*(D_1) \ge \Phi(D_1)$ and $\inf_{D\ge0} \{ P_*(D)-\Phi(D) \} \ge 0$. \\
Therefore, $P_*(D) \ge \Phi(D)$ and the proof is complete.
\end{proof}

We end this subsection by calculating the bubble size, $B(D) = P_*(D)-I(D)$, explicitly. Using the formula of $\Phi(D)$ and $I(D)$, $B(D)$ can be written as 
\begin{equation}\label{bubble1}
\begin{aligned}
B(D) = \begin{cases}
Em(D) & \text{if } D \le \bar{D}\\
-\frac{(\kappa_1-\kappa_2)(D-\bar{D})}{(\lambda+\kappa_1)(\lambda+\kappa_2)} + Em(D) & \text{if } \bar{D} <D \le \tilde{D} \\
Fu(D) & \text{if } D>\tilde{D}\\
\end{cases}.
\end{aligned}
\end{equation}
Note that $\bar{D}$ may be negative and for this, there are only two cases for $B(D)$. $\tilde{D}$ is always positive because we have assumed that $\kappa_1 \theta_1 > \kappa_2 \theta_2$.

One may be curious about the exact value of $B(D)$, as the values of the confluent hypergeometric functions are not apparent. We provide a visualization of the bubble size for some sample parameters in Section 5.3. We can see that the bubble is most apparent when $D$ is near zero and quickly disappears when $D > \tilde{D}$.

\subsection{Case 2: \texorpdfstring{$\theta_1 >  \theta_2$}{Case 2: theta 1 > theta 2}}

Similar to the previous section, we attempt to calculate the bubble size explicitly. It would be ideal if the function $\Phi$ defined in Theorem 5.1 was still the minimal equilibrium price. However, there is a small issue in Theorem 5.1 in that, we cannot guarantee the positiveness of the constant $E$. $E$ and $F$ must be positive for $\Phi$ to be an equilibrium price, because $\Phi(D)$ must be at least $I(D)$. Although the positiveness of $F$ trivially follows from the assumption $\theta_1 > \theta_2$, $E$ is not necessarily positive (see the Appendix). In fact, for extreme cases when $\theta_1 \gg \theta_2$ and $b_2 \gg x_2$, $E$ becomes negative. In this section, we first calculate the bubble size under the assumption that $E \ge 0$. Next, we provide a lower bound for the bubble size for the general case.

First, we assume that $E \ge 0$. Random simulations on the parameter values indicate that this assumption is valid in most cases. Under the assumption that $E \ge 0$, Theorem 5.1 and Lemma 5.2 hold. Hence, $\Phi$ is the minimal equilibrium price by Corollary 5.3. In this case, $\bar{D} > \tilde{D}$, and the formula of $B(D)$ is different from the one in Section 5.1. More precisely, the bubble size changes to
\begin{equation}\label{bubble2}
\begin{aligned}
B(D) = \begin{cases}
Em(D) & \text{if } D \le \tilde{D}\\
\frac{(\kappa_1-\kappa_2)(D-\bar{D})}{(\lambda+\kappa_1)(\lambda+\kappa_2)} + Fu(D) & \text{if } \tilde{D} <D \le \bar{D} \\
Fu(D) & \text{if } D>\bar{D}\\
\end{cases}.
\end{aligned}
\end{equation}

Next, we consider the general case in which $E$ may be negative. In this case, $\Phi(D) < I(D)$ for $D \le \tilde{D}$ and $\Phi(D)$ is not an equilibrium price. However, $\Phi(D)$ is still a $C^2$ solution of \eqref{visc} with linear growth at infinity and $P_*$ is a viscosity supersolution of \eqref{visc}. Therefore, we can use the proof of Corollary 5.3 to show that $P_*(D) \ge \Phi(D)$. Hence,
$$B(D) = P_*(D)-I(D) \ge \Phi(D)-I(D)$$
for $D > \tilde{D}$. We can also calculate the lower bound when $D \le \tilde{D}$ by finding another positive increasing $C^2$ solution of \eqref{de}. However, the choice of this solution depends heavily on the initial parameters, and we omit the details. In the following remark, we analyze the behavior of the bubble, using the formula of $B(D)$ in \eqref{bubble1} and \eqref{bubble2}.

\begin{remark}
Our bubble formula satisfies the following facts.
\begin{enumerate}[(a)]
	\item The bubble disappears as $\theta_1 - \theta_2 \to 0$ and $\kappa_1 - \kappa_2 \to 0$, i.e. as the heterogeneous parameters converge to common values.
	\item The bubble increases as $\kappa_1 - \kappa_2$ increases,  i.e. as the difference between the heterogeneous beliefs increases.
	\item A bubble does not exist when $\kappa_1=\kappa_2$ and $\theta_1\neq \theta_2,$ i.e. heterogeneous beliefs may not cause a bubble. 
	\item A bubble exists when $\kappa_1>\kappa_2$ and $\theta_1 = \theta_2,$ i.e. different beliefs on the mean-reversion rate with common beliefs on the mean-reversion level cause a bubble.
	\item The bubble becomes larger as $\sigma$ increases, i.e. as the volatility increases.
\end{enumerate}  
\end{remark}

These results rely on the assumption that $E \ge 0$ and can be verified by the formula \eqref{bubble1} and \eqref{bubble2}. Recall that $B(D)$ is expressed as $Em(D)$ or $Fu(D)$, conditional to the initial dividend rate, $D$. As $E$ and $F$ are linear with respect to $\theta_1 - \theta_2$ and $\kappa_1 - \kappa_2$, $B(D)$ is also linear in these terms (see the formula of $E$ and $F$ in \eqref{E} and \eqref{F}). Results (a) and (b) directly follows from this observation. 
Results (c) and (d) comes from Theorem \ref{thm:main}.
Result (d) also can be checked by plugging $\theta_1 = \theta_2$ in \eqref{E} and \eqref{F} to see that $E, F > 0$. Result (e) follows from the fact that $E$ and $F$ are increasing functions of $\sigma$.

Recall from Section 2 that, at time $t$, the asset would be possessed by the investor that evaluates \eqref{ivalue} highly. In fact, we can identify the owner of the asset at each $t$ through simple calculations. The short answer is that group $1$ holds the asset if $D_t < \tilde{D}$, and group $2$ possesses the asset if $D_t > \tilde{D}$.\footnote{This result is consistent with the optimal strategy from Theorem 2.3 of \cite{muhle2018risk}.} We note that this transaction threshold value, $\tilde{D}$, is also determined by the drift term in \eqref{D_t}. Indeed, a simple comparison on the drift term shows that
$$\kappa_1 (\theta_1 - D) > \kappa_2 (\theta_2 - D)$$
if and only if $D < \tilde{D}$. The proof of the claim directly follows from this comparison and the fact that $P_*$ satisfies \eqref{de'} and we omit the details.

Based on this observation, we can explain behavior (e). As $\sigma$ increases, the process $D_t$ crosses $\tilde{D}$ more often. Hence, investors have more opportunities for trading, and this leads to a larger bubble. This explanation is consistent with the following statement from \citet{scheinkman2003overconfidence}, even though our approaches differ: ``bubbles are accompanied by large trading volume and high price volatility".

\begin{remark}
Our result about the ownership of the asset is consistent with the optimal strategy from Theorem 2.3 of \cite{muhle2018risk}. Under a market clearing assumption, they claim that the asset is entirely possessed by the investor who achieves the maximum in the supremum term of a differential equation. In the case of our problem setting, this differential equation becomes \eqref{de} and the optimal strategy coincides. However, it is not clear whether the market clearing price is equivalent to our equilibrium definition in the general setting. Our price scheme only considers the first possible transaction whereas there may be uncountable continuous time trades according to the market clearing price. Also, typically there are infinitely many equilibrium prices (this is why we focused on the minimal price), but the market clearing price in \cite{muhle2018risk} tends to be unique.
\end{remark}

\subsection{Visualization and interpretation of the bubble}
Now, we visualize the aforementioned bubble formula. We compare the bubbles created from three different parameter values ($\theta_1 > \theta_2$, $\theta_1 < \theta_2$, and $\theta_1 = \theta_2$). It is apparent that the bubble size is different for each case.

First, we present an example of the intrinsic value and the minimal equilibrium price when $\theta_1 < \theta_2$ and $\kappa_1\theta_1 \ge \kappa_2\theta_2$ in Figure 1. Here, we set the parameter values as 
\begin{equation}\label{initial2}
\begin{aligned}
\kappa_1 =0.2, \kappa_2 = 0.1, \theta_1 = 0.015, \theta_2 = 0.02, \lambda = 0.02, \sigma = 0.02.
\end{aligned}
\end{equation}
These initial values were chosen to be similar to those of \citet{chen2011asset} and the mean level beliefs were set to realistic values of $0.015$ and $0.02$. The discount rate $\lambda$ can also be interpreted as the interest rate and was set to $0.02$. For these initial values, $\bar{D} = -0.04<0$, and the intrinsic value is a linear function. It should be noted that the bubble is the largest at $D=0$ and rapidly disappears as $D$ increases. In particular, the bubble is nearly indistinguishable after $D>\tilde{D}=0.01$.
\begin{figure}[ht]
    \centering
    \includegraphics[width=0.5\linewidth]{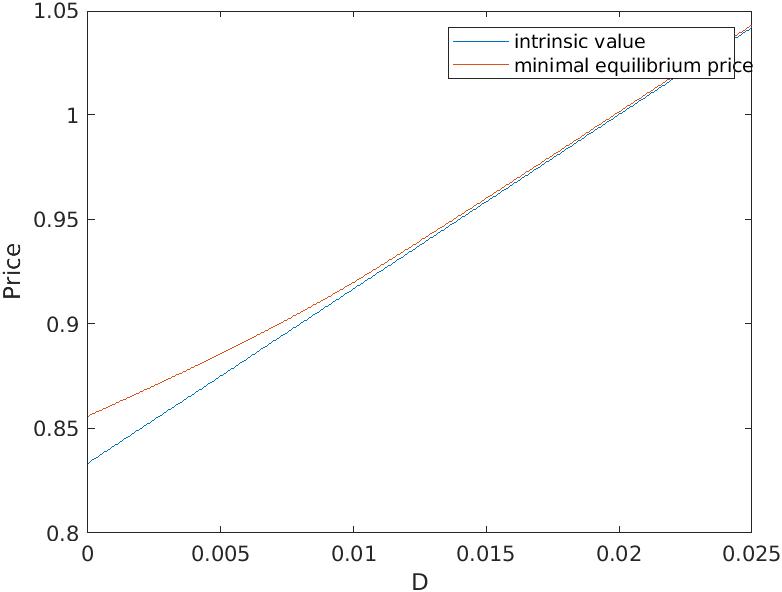}
    \caption{Minimal equilibrium price and the intrinsic value as a function of $D$ for the initial values in \eqref{initial2}}
\end{figure}

We also see how the relative size of the bubble change with respect to the initial dividend rate. We define $R(D):= P_*(D)/I(D)-1$ as the relative bubble size. In Figure 2, we see that $R(D)$ is 2.71$\%$ when $D=0$ and linearly decreases to 0.35$\%$ as $D$ increases to $\tilde{D}=0.01$. After $\tilde{D}=0.01$, $R(D)$ converges to 0. We see that although the bubble mathematically exists, its size is very small and is most apparent at low initial dividend rates.

\begin{figure}[ht]
    \centering
    \includegraphics[width=0.5\linewidth]{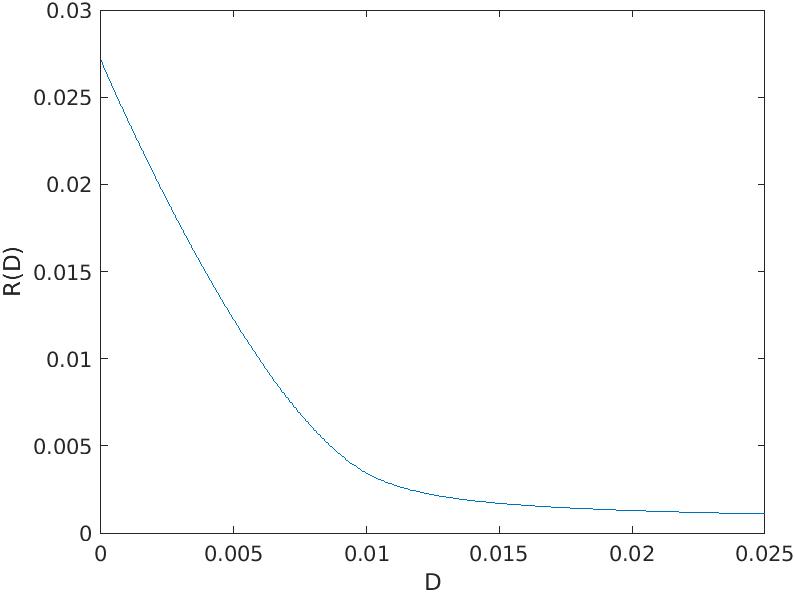}
    \caption{Relative size of the bubble for the parameters in \eqref{initial2}}
\end{figure}

Next, we consider the case in which $\theta_1 > \theta_2$. Here, we set the parameter values as 
\begin{equation}\label{initial1}
\begin{aligned}
\kappa_1 =0.2, \kappa_2 = 0.1, \theta_1 = 0.04, \theta_2 = 0.02, \lambda = 0.02, \sigma = 0.02.
\end{aligned}
\end{equation}
For these values $E = 1.15 \times 10^{-4}>0$ and $P_* = \Phi$ due to the reasoning in Section 5.2. We can see in Figure 3 that market transactions smooth the equilibrium price, in contrast to the non-differentiable intrinsic value. Additionally, the bubble size when $D=0$ is almost 0 and increases until $\bar{D} = 0.26$. This result sharply contrasts the results in \eqref{initial2} and we see that the sign of $\theta_1-\theta_2$ significantly contributes to the bubble size. Even though it is not apparent in Figure 3, the bubble size converges to 0 as $D$ goes to infinity. 
\begin{figure}[ht]
    \centering
    \includegraphics[width=0.5\linewidth]{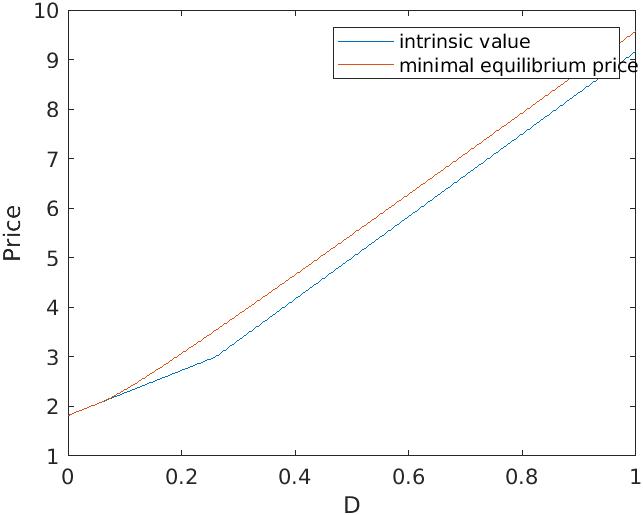}
    \caption{Minimal equilibrium price and the intrinsic value as a function of $D$ for the initial values in \eqref{initial1}}
\end{figure}

We also provide the relative bubble size $R(D)$ in Figure 4. $R(D)$ is close to 0 when $D \le \tilde{D} = 0.06$ and increases to almost 18$\%$ as $D$ increases to $\bar{D} = 0.26$. We note that the price bubble is most apparent around $\tilde{D}$ and disappears as $D \to \infty$.
\begin{figure}[ht]
    \centering
    \includegraphics[width=0.5\linewidth]{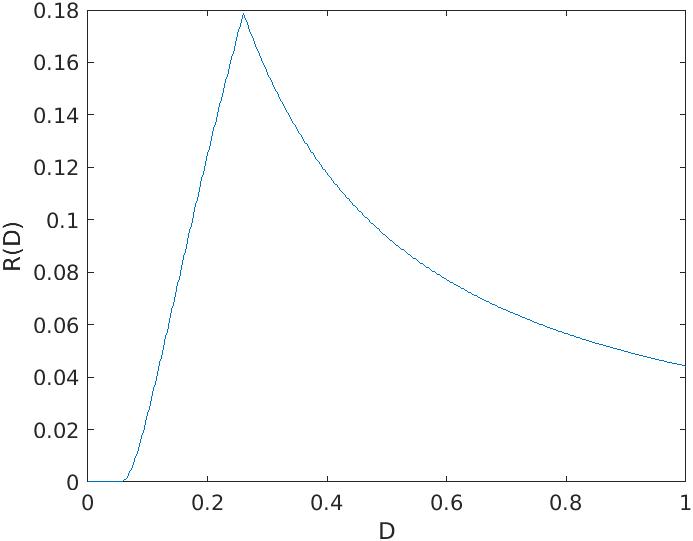}
    \caption{Relative size of the bubble for the parameters in \eqref{initial1}}
\end{figure}
The comparison of bubbles from the initial values \eqref{initial2} and \eqref{initial1} provides some interesting results. Although the minimal equilibrium price is defined by the same function $\Phi$, the bubble size is heavily dependent on the parameter values. With $\theta_1 < \theta_2$, the bubble is significant when $D < \tilde{D}$. However, when $\theta_1 > \theta_2$, the bubble size is nearly 0 when $D < \tilde{D}$, and is most apparent when $\tilde{D} < D < \bar{D}$. 

Finally, we consider the case when $\theta_1 = \theta_2$. In other words, the investors disagree only on the rate of mean-reversion. Hence, this case is basically a positive modification of the Chen--Kohn model. Figures 5 and 6 are based on the initial parameters
\begin{equation}\label{initial3}
\begin{aligned}
\kappa_1 =0.2, \kappa_2 = 0.1, \theta_1 = \theta_2 = 0.04, \lambda = 0.02, \sigma = 0.02.
\end{aligned}
\end{equation}

\begin{figure}
    \centering
    \includegraphics[width=0.5\linewidth]{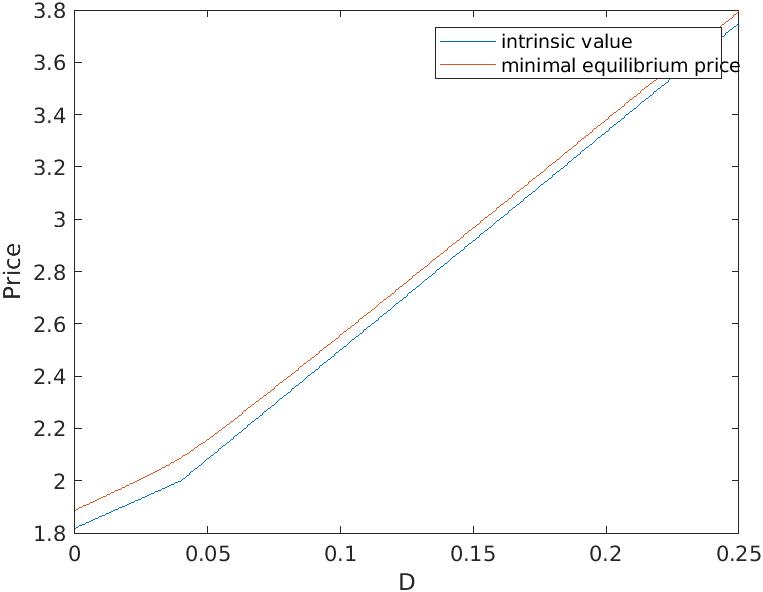}
    \caption{Minimal equilibrium price and the intrinsic value as a function of $D$ for the parameters in \eqref{initial3}}
\end{figure}

\begin{figure}
    \centering
    \includegraphics[width=0.5\linewidth]{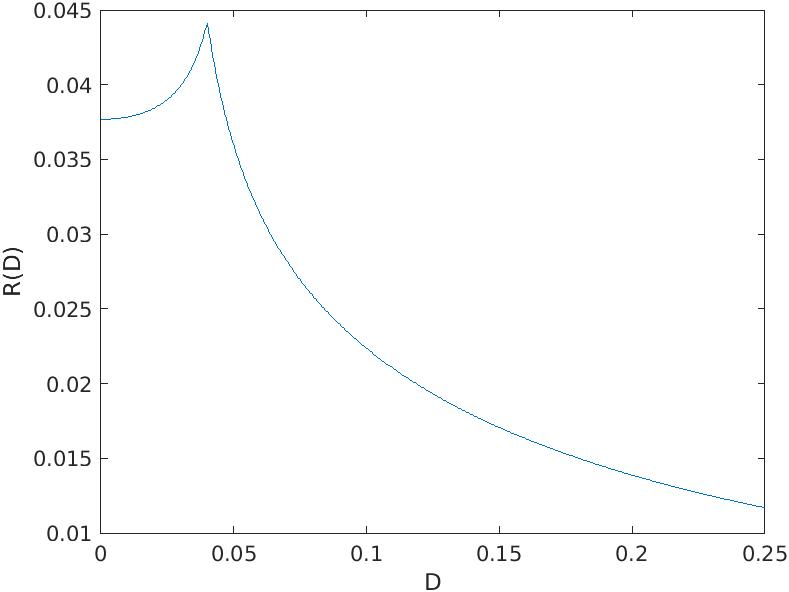}
    \caption{Relative size of the bubble for the parameters in \eqref{initial3}}
\end{figure}

\noindent Note that $\bar{D}=\tilde{D}$ when $\theta_1 = \theta_2$. Hence, the formula for the bubble size changes only at this value. Here, the bubble is apparent for all values of $D$ and can be regarded as a compromise between the case when $\theta_1 < \theta_2$ and $\theta_1 > \theta_2$. Indeed, we can see that Figure 5 is similar to Figure 1 from \citet{chen2011asset}.

\section{Conclusion}
In this study, we proved that price bubbles exist if and only if $\kappa_1 > \kappa_2$ and $\kappa_1\theta_1 > \kappa_2 \theta_2$. In other words, the bubbles do not exist if one group uniformly dominates the drift term of the dividend stream dynamics. Heterogeneous beliefs on volatility have no effect on the formation of bubbles, even though they have a certain impact on the bubble size. These results emphasize that heterogeneous beliefs do not always result in a bubble, and we need to consider the drift term of the model.

We also calculated the bubble size when $\sigma_1 = \sigma_2$. The minimal equilibrium price is identified as a classical solution of a differential equation, and the bubble size can be written in terms of special functions. This bubble displays intuitive behaviors, and we have provided some numerical examples for its visualization. Similar to \citet{chen2011asset}, an existing bubble is permanent if our investors have consistent heterogeneous beliefs.

Finally, we mention that there are many rooms left for further research. Although we used an infinite time horizon without transaction costs, the extension to the finite time horizon with transaction costs is straightforward and we anticipate similar behaviors. Also, our calculations on the bubble size and the corresponding trading strategy were done under the assumption that $\sigma_1 = \sigma_2$, and analysis for the general case is left for subsequent research.

$$ $$

\noindent\textbf{Acknowledgments.}\\ 
Hyungbin Park was supported by the National Research Foundation of Korea (NRF) grants funded by the Ministry of Science and ICT (No. 2017R1A5A1015626, No. 2018R1C1B5085491 and No. 2021R1C1C1011675) 
and the Ministry of Education   (No. 2019R1A6A1A10073437) through the Basic Science Research Program.

\appendices
\section{Appendix: explicit formula of \texorpdfstring{$E, F$}{E, F}}
The explicit value of the constants $E, F$ from Section 5 can be expressed by the following formula. We recall that $E, F$ are constants that depend on the model parameters, $\kappa_1, \kappa_2, \theta_1, \theta_2, \lambda, \text{ and } \sigma$.
\begin{equation}\label{E}
E = \frac{1}{A}\left[u_2(\kappa_1-\kappa_2) - \frac{2U_2 \kappa_1 \kappa_2 (\theta_1 - \theta_2)}{\sigma^2} \right].
\end{equation}
\begin{equation}\label{F}
F = \frac{1}{A}\left[M_1\kappa_2\frac{\theta_1-\theta_2}{\theta_1} + m_1(\kappa_1-\kappa_2)\right].    
\end{equation}
Here, $m_1:=M(a_1,b_1, x_1), u_2:=U(a_2, b_2, x_2), M_1:= M(a_1+1,b_1+1,x_1), U_2:=U(a_2+1,b_2+1,x_2)$, where $a_i, b_i$ are the constants defined in Theorem 5.1 and $x_i = \frac{2\kappa_i \tilde{D}}{\sigma^2}$. Recall that the functions $M$ and $U$ are confluent hypergeometric functions of the first and second type, respectively. $A$ is a constant defined as
$$A:= \frac{2m_1U_2}{\sigma^2}+\frac{M_1u_2}{\kappa_1\theta_1}.$$
As $m_i, u_i, M_i, U_i$ are positive, $A$ is clearly positive. From \eqref{E} and \eqref{F}, we can see that $E$ and $F$ are both positive when $\theta_1 = \theta_2$.

Now, we prove that $E$ and $F$ are positive under the assumptions of Theorem 5.1. We assume that $\kappa_1>\kappa_2$ and $\theta_1 \le \theta_2$. The formula of $E$ in \eqref{E} shows that $E>0$. Unfortunately, the proof of $F>0$ requires more complicated calculations, and we prove this using continued fractions. By the formula of $F$ in \eqref{F}, it suffices to show that
\begin{equation}
\begin{aligned}
\frac{M(a_1,b_1,x_1)}{M(a_1+1,b_1+1,x_1)} \ge \frac{\kappa_2 (\theta_2 - \theta_1)}{(\kappa_1-\kappa_2)\theta_1}.
\end{aligned}
\end{equation}

\noindent Note that the ratio of confluent hypergeometric functions of the first type can be expressed by continued fractions. We use the formula
\begin{equation}
\begin{aligned}
\frac{M(a_1,b_1,x_1)}{M(a_1+1,b_1+1,x_1)} = \frac{b_1-x_1}{b_1} + \frac{1}{b_1}{\mathop{\huge\mathrm{K}}_{m=1}^\infty}\frac{mx_1}{b_1+m-x_1} 
\end{aligned}
\end{equation}
from \citet{cuyt2008handbook}. Here, the symbol $\mathop{\mathrm{K}}$ is defined as
$${\mathop{\mathrm{K}}_{m=1}^\infty}\frac{c_m}{d_m} = \frac{c_1}{d_1 + \frac{c_2}{d_2 + \frac{c_3}{d_3+ \ddots}}}.$$
Because $b_1-x_1 = \frac{2\kappa_1(\theta_1 - \tilde{D})}{\sigma^2} \ge 0$, ${\mathop{\huge\mathrm{K}}_{m=1}^\infty}\frac{mx_1}{b_1+m-x_1}$ is positive. Meanwhile, 
\begin{equation}
\frac{b_1-x_1}{b_1} = \frac{\theta_1 - \tilde{D}}{\theta_1}
= \frac{(\kappa_1-\kappa_2)\theta_1 - (\kappa_1\theta_1-\kappa_2\theta_2)}{(\kappa_1-\kappa_2)\theta_1}
= \frac{\kappa_2 (\theta_2 - \theta_1)}{(\kappa_1-\kappa_2)\theta_1}
\end{equation}
and the proof is complete.

As stated in Section 5.2, $E$ may be negative in certain extreme cases. We now state the equivalent condition for $E \ge 0$. By the formula of $E$ in \eqref{E}, this is equivalent to
\begin{equation}\label{appendixE}
\begin{aligned}
(\kappa_1-\kappa_2)U(a_2,b_2,x_2) \ge \frac{2\kappa_1\kappa_2(\theta_1-\theta_2)}{\sigma^2} U(a_2+1, b_2+1, x_2).
\end{aligned}
\end{equation}

\noindent As the left hand side of \eqref{appendixE} is nonnegative, it is sufficient to prove \eqref{appendixE} for $\theta_1>\theta_2$.
By rewriting \eqref{appendixE}, this is equivalent to
\begin{equation}\label{Uratio}
\begin{aligned}
\frac{U(a_2+1, b_2+1, x_2)}{U(a_2,b_2,x_2)} \le \frac{(\kappa_1-\kappa_2)\sigma^2}{2\kappa_1\kappa_2(\theta_1-\theta_2)}
= \frac{1}{x_2-b_2}.
\end{aligned}
\end{equation}
We can take $a_2=1,b_2=4.5,x_2=1000$ to see that this is not always true.

\bibliographystyle{plainnat}
\bibliography{reference}
\end{document}